\documentclass[12pt,english]{iopart} 
\usepackage[utf8]{inputenc}
\usepackage[usenames,dvipsnames]{color}
\usepackage[normalem]{ulem}
\usepackage{graphicx}

\expandafter\let\csname equation*\endcsname\relax
\expandafter\let\csname endequation*\endcsname\relax

\usepackage{amsmath}
\usepackage{amssymb,amsthm}
\usepackage{amsfonts}
\usepackage{color}
\usepackage{graphics}
\usepackage{epsfig}
\usepackage{bbm}
\usepackage{pstricks}
\usepackage{subfigure}
\usepackage{bm}
\usepackage{bbold}
\usepackage{cite}
\usepackage{braket}

\newcommand{\eins}{\mathbbm{1}}

\newcommand{\bea}{\begin{eqnarray}}
\newcommand{\eea}{\end{eqnarray}}
\newcommand{\be}{\begin{equation}}
\newcommand{\ee}{\end{equation}}
\newcommand{\BraKet}[2]{\ensuremath{\langle #1|#2\rangle}}
\newcommand{\WW}{{\mathcal{W}}}
\newcommand{\ketbra}[1]{\ensuremath{| #1 \rangle\!\langle #1 |}}

\newtheorem{theorem}{Theorem}

\newtheorem{definition}[theorem]{Definition}

\begin{document}

\title[Characterizing multipartite entanglement 
classes]{Characterizing multipartite entanglement 
classes via higher-dimensional embeddings}

\author{Christina Ritz, Cornelia Spee, and Otfried G\"uhne}

\address{Naturwissenschaftlich-Technische Fakult\"at,
Universit\"at Siegen,\\
Walter-Flex-Stra{\ss}e~3,
57068 Siegen, Germany}


\date{\today}

\begin{abstract} 
Witness operators are a central tool to detect entanglement or to distinguish 
among the different entanglement classes of multiparticle systems, which can 
be defined using stochastic local operations and classical communication (SLOCC). 
We show a one-to-one correspondence between general SLOCC witnesses and a class
of entanglement witnesses in an extended Hilbert space. This relation can be used 
to derive SLOCC witnesses from criteria for full separability of quantum states;
moreover, given SLOCC witnesses can be viewed as entanglement witnesses. As 
applications of this relation we discuss the calculation of overlaps between 
different SLOCC classes and the SLOCC classification in $2\times 3\times 3$-dimensional 
systems.
\end{abstract}

\section{Introduction}

Entanglement is considered to be an important resource for applications 
in quantum information processing, making its characterization 
essential for the field \cite{hororeview, gtreview}. This includes its quantification 
and the development of tools to distinguish between different classes of 
entanglement. In general, entanglement is a resource if the parties 
are spatially separated and therefore the allowed operations are 
restricted to local operations assisted by classical communication 
(LOCC). It can neither be generated nor increased by LOCC transformations. 
Hence, convertibility via LOCC imposes a partial order on the entanglement 
of the states, and this order has been studied in detail \cite{Nielsen, 
kleinmann, ChLe12, MES, connylang, trivialsym}.

For multipartite systems the classification via LOCC is, however, even for pure
states very difficult, so one may consider a coarse grained classification. This 
can be done using the notion of stochastic local operations assisted by classical 
communication (SLOCC).  By definition, an SLOCC class is formed by those pure states 
that can be converted into each other via local operations and classical communication 
with non-zero probability of success \cite{SLOCC222}. SLOCC classes and their transformations
have been characterized for small system sizes and symmetric states \cite{SLOCC222, SLOCC2222, 
SLOCC333, ChMi10, hebenstreit, bastin, liduan, steinhoff, sawicki, sanz}
and it has been shown that for multipartite systems 
there are finitely many SLOCC classes for tripartite systems with local dimensions of 
up to $2\times 3\times m$ and infinitely many otherwise 
\cite{SLOCCfinite}.

Another important problem in entanglement theory is the separability problem, 
i.e., the task to decide whether a given quantum state is entangled or 
separable. Even though several criteria have been found which can decide 
separability in many instances \cite{hororeview, gtreview, peresppt, sep1, 
sep2, sep3, rudolph, doherty}, 
the question whether a general multipartite mixed state is 
entangled or not, remains highly non-trivial. In fact, if the separability 
problem is formulated as a weak membership problem, it has been proven to be 
computationally NP-hard \cite{sepNP, gharibian} in the dimension of the system.

One method to certify entanglement uses entanglement witnesses \cite{EW, EW1, gtreview}. 
An entanglement witness is a hermitian operator which has a positive 
expectation value for all separable states but gives a negative value for at 
least one entangled state. In opposition to other criteria, one main advantage 
of witnesses lies in the fact that no complete knowledge of the state is necessary 
and one just has to measure the witness observable. A special type of witnesses 
are projector-based witnesses of the form $\WW= \lambda  \eins - \ketbra{\psi}$, 
with $\lambda$ being the maximal squared overlap between the entangled state 
$\ket{\psi}$ and the set of all product states. Such projector based witnesses 
can also be used to distinguish between different SLOCC classes 
\cite{3qubitmixed, symmetricmixedSLOCC}. In that case, $\lambda$ is the maximal 
squared overlap between a given state $\ket{\psi}$ in SLOCC class $S_{\ket{\psi}}$
and the set of all states within another SLOCC class $S_{\ket{\varphi}}$.  
If a negative expectation value of $\WW$ is measured, the considered state
$\varrho$ cannot be within the convex hull of $S_{\ket{\varphi}}$ or lower 
entanglement classes. In this context one should note that such statements
require an understanding of the hierarchic structure of SLOCC classes, in the
sense that some classes are contained in others \cite{3qubitmixed, symmetricmixedSLOCC}.

In this paper we establish an one-to-one correspondence between general SLOCC 
witnesses for multipartite systems and a class of entanglement witnesses in a
higher-dimensional system, built by two copies of the original one. This extends 
the results of Ref.~\cite{hulpke} from the bipartite setting to the multipartite 
one and provides at the same time a simpler proof. The equivalence provides not 
only a deeper insight in the structure of SLOCC classes but enables to construct 
whole sets of entanglement witnesses for high-dimensional systems from the SLOCC 
structure of lower dimensions and vice versa. As such, from the solution for one 
problem, the solution to the related one readily follows. 

The paper is organized as follows. In Section II we briefly review the notion of 
SLOCC operations, entanglement witnesses and SLOCC witnesses. Section III states 
the main result of our work, the one-to-one correspondence among certain entanglement- 
and SLOCC witnesses. 
Furthermore, as optimizing the overlap $\lambda$ between SLOCC classes is in 
general a hard problem and as such often not feasible analytically, a possible 
relaxation of the set of separable states to states with positive partial 
transpose is discussed. Section IV focuses on systems consisting of one qubit 
and two qutrits. Using numerical optimization, we find the maximal overlaps 
between all pairs of representative states of one SLOCC class and arbitrary states 
of another SLOCC class. The implications of these results for the hierarchic 
structure of SLOCC classes are then discussed. Section V concludes the paper 
and provides an outlook.


\section{Preliminaries}
In this section the basic notions and definitions are briefly 
reviewed. We start with the notion of SLOCC equivalence of two 
states and then move on to the definition of entanglement witnesses. 
Finally, we will relate both concepts by recapitulating the notion 
of witness operators that are able to separate between different SLOCC 
classes.

\subsection{SLOCC classes}
As mentioned before two pure states are within the same SLOCC class if 
one can convert them into each other via LOCC with a non-zero probability 
of success. It can be shown that this implies the following definition 
\cite{SLOCC222}. 

\begin{definition}
Two $N$-partite pure quantum states $\ket{\psi}$, $\ket{\varphi}$ are called equivalent 
under SLOCC if  there are $N$ matrices $\{A_i |~ \det(A_i)\neq 0\}$ 
such that
 \begin{equation}
  \ket{\varphi} = \bigotimes_{i=1}^N A_i \ket{\psi} \mbox{ and consequently } 
  \ket{\psi} = \bigotimes_{i=1}^N A^{-1}_i \ket{\varphi}.
 \end{equation}
 \end{definition}
That is, an SLOCC class or SLOCC orbit includes all states that are related by 
local, invertible operators. To extend this definition to mixed states one 
defines the class $S_{\ket{\psi}}$ with the representative $\ket{\psi}$ as 
those mixed states that can be built as convex combinations of pure states 
within the SLOCC orbit of $\ket{\psi}$ and of all pure states that can be 
approximated arbitrarily close by states within this orbit 
\cite{3qubitmixed, symmetricmixedSLOCC}.

\subsection{Entanglement witness}
An hermitean operator that can be used to distinguish between different 
classes of entanglement is called a witness operator. Recall that a mixed 
state that can be written as a convex combination of product states of the 
form $\ket{\psi_s} = \ket{A}\ket{B} \cdots \ket{N}$ is called fully separable, 
and states which are not of this form are entangled \cite{hororeview, gtreview}. A witness 
operator that can certify entanglement has to fulfill the following properties 
\cite{EW, EW1}:

\begin{definition} A hermitean operator $\WW$ is an entanglement witness if 
\begin{align}
(i)& ~\tr(\varrho_{s} \WW)\geq 0 \mbox{ for all separable states } \varrho_{s}, \nonumber \\
(ii)& ~\tr(\varrho_{e} \WW) < 0 \mbox{ for at least one state } \varrho_e,
\end{align}
holds.
\end{definition}
Hence, $\WW$ witnesses the non-membership with respect to the convex set of separable 
states. If $\tr(\varrho \WW)<0$ for some state $\varrho$, then $\WW$ is said to detect 
$\varrho$. A special class of witness operators are projector-based witnesses. Their 
construction is based on the maximal value $\lambda$ of the squared overlap between 
a given entangled state $\ket{\psi}$ with the set of all product states 
$\{\ket{\psi_{s}}\}$. More precisely, $\WW= \lambda  \eins - \ketbra{\psi}$  
with $\ket{\psi}$ being some entangled state and 
$\lambda= \sup_{\{\ket{\psi_{s}}\}} |\BraKet{\psi}{\psi_s}|^2$ is a valid entanglement 
witness \cite{gtreview}. We stress, however, that this is not the most general
way to construct witnesses. Other construction methods make use of various 
separability criteria or physical quantities like Hamiltonians in spin models 
\cite{gezaspin} or structure factors \cite{hermann}.

\subsection{SLOCC witness}
The concept of entanglement witnesses can be generalized to SLOCC witnesses. An 
SLOCC witness is an operator from which one can conclude that a state $\varrho$ 
is not in the SLOCC class $S_{\ket{\psi}}$ \cite{3qubitmixed, symmetricmixedSLOCC}. 

\begin{definition}
A hermitean operator $\WW$ is an SLOCC witness for the class
$S_{\ket{\psi}}$ if
\begin{align}
(i)& ~\tr(\ketbra{\eta} \WW) \geq 0 \mbox{ for all pure states $\ket{\eta} \in S_{\ket{\psi}}$},
\nonumber \\
(ii)& ~\tr(\varrho \WW) < 0 \mbox{ for at least one state $\varrho$}, 
\end{align}
holds. 
\end{definition}
Thus $\WW$ detects for $\tr(\varrho \WW) < 0$ states $\varrho$ that are not within 
$S_{\ket{\psi}}$. Note that it suffices to check positivity on all pure states 
$\ket{\eta}$  in the set of mixed states $S_{\ket{\psi}}$, as these form the extreme
points of this set. Also if one considers $\ket{\psi} = \ket{A}\ket{B} \cdots \ket{N},$
then the set of all SLOCC equivalent states are just all product states and the SLOCC
witness is just a usual entanglement witness. 

One can construct an SLOCC witness via 
\begin{equation}
\WW= \lambda \eins - \ketbra{\varphi},
\end{equation}
where $\lambda$ denotes the maximal squared overlap between all pure 
states $\ket{\eta}$ in the SLOCC class $S_{\ket{\psi}}$ and the 
representative state $\ket{\varphi}$ of SLOCC class $S_{\ket{\varphi}}$, 
i.e. $\lambda = \sup_{\{\ket{\eta}\}}{|\braket{\varphi|\eta}|^2}$. Our main
result, however, does not assume this type of witness and is valid for general
SLOCC witnesses.

A special class of SLOCC witnesses are those verifying the Schmidt rank of a 
given bipartite state. Note that the Schmidt rank is the only SLOCC invariant 
for bipartite systems, and a one-to-one correspondence between Schmidt 
number witnesses and entanglement witnesses in an extended Hilbert space has 
been found \cite{hulpke}. 
In the next section we will show that in fact there 
is a one-to-one correspondence between SLOCC- and entanglement witnesses 
for arbitrary multipartite systems.

\section{One-to-one correspondence between SLOCC- and entanglement witnesses}
\label{secEWSLOCCW}
In the following we will show how to establish a one-to-one 
correspondence between SLOCC witnesses and certain entanglement 
witnesses within a higher-dimensional Hilbert space for arbitrary 
multipartite systems. In order to improve readability, 
our method will be presented for the case of tripartite systems, however, the 
generalization to more parties is straightforward. Then, we will discuss one 
possibility to use this correspondence to derive an SLOCC witnesses from separability
criteria.

\subsection{The correspondence between the two witnesses}
Let us start with formulating the problem. Consider the pure state $\ket{\psi}$, 
which is a representative state of the SLOCC class $S_{\ket{\psi}}$. Then all pure states,
$\ket{\eta}$, within the SLOCC orbit of $\ket{\psi}$ can be reached by applying local 
invertible operators $A,B$ and $C$, that is $\ket{\eta}=A \otimes B \otimes C \ket{\psi}.$
Here, one has to take care that $\ket{\eta}$ is normalized; so, if considering general matrices
$A,B,C$, one has to renormalize the state.
The aim will be to maximize the overlap between a given state $\ket{\varphi}$ and a 
pure state $\ket{\eta}$ within $S_{\ket{\psi}}$,
\begin{equation}
 \sup_{\ket{\eta} \in S_{\ket{\psi}}} |\braket{\varphi|\eta}|
 =
 \sup_{A,B,C}{\frac{|\braket{\varphi|A \otimes B\otimes C|\psi}|}{\Vert A\otimes B\otimes C \ket{\psi}\Vert}},
\end{equation}
which is the main step for constructing the projector-based witness. Stated differently, 
the quantity of interest is the minimal value $\lambda > 0$, such that
\begin{equation}
\label{1}
\sup_{A,B,C}{\frac{ |\braket{\varphi|A \otimes B\otimes C|\psi}|}{\Vert A\otimes B\otimes C \ket{\psi}\Vert}} \leq \sqrt{\lambda}.
 \end{equation}
It can easily be seen that this is true if and only if
\begin{align}
\lambda & \braket{\psi|A^{\dagger}A\otimes B^{\dagger}B\otimes C^{\dagger}C|\psi} \nonumber \\
& - \braket{\psi|A^{\dagger}\otimes B^{\dagger}\otimes C^{\dagger} 
\ket{\varphi}\bra{\varphi}A\otimes B\otimes C|\psi} \geq 0
 \label{3}
 \end{align}
 holds.
 One can then define a witness operator $\WW=\lambda  \eins- \ketbra{\varphi}$ which, 
with the definition of $\ket{\eta}$ from before, satisfies:
\begin{equation}
\label{4}
\braket{\eta|\WW|\eta} \geq 0.
\end{equation}
Note that in the formulation of Eqs.~(\ref{3}, \ref{4}) the normalization of 
$\ket{\eta}=A \otimes B \otimes C \ket{\psi}$ is irrelevant, this trick has 
already been used in Ref.~\cite{kampermann}.

The key idea to establish the connection is the following: In order to prove 
that $\WW$ is an SLOCC witness, one has to minimize in Eq.~(\ref{3}) over all 
matrices $A,B,C$, which do not have any constraint anymore. A matrix like $A$
acting on the Hilbert space $\mathcal{H}_A$ can be seen as a vector on the 
two-copy system $\mathcal{H}_{A_1}\otimes \mathcal{H}_{A_2}$. Then, the 
remaining optimization is the same as optimizing over all product states in 
the higher-dimensional system and requesting that the resulting value is always 
positive. Consequently, the SLOCC witness $\WW$ corresponds to a usual witness 
$\tilde{\WW}$ on the higher-dimensional system. More precisely, as stated in 
the following theorem, one can show that if Eq.~(\ref{4}) holds, then the 
operator $\tilde{\WW}=\WW \otimes \ketbra{\psi^*}$ is positive on all separable 
states $\ket{\xi_{\rm sep}}$ and vice versa. Here and in the following $^*$ denotes 
complex conjugation in a product basis.

\begin{theorem}
\label{Th.1}
Consider the operator $\WW$  on the tripartite 
space $\mathcal{H}=\mathcal{H}_{A}\otimes \mathcal{H}_{B} \otimes \mathcal{H}_{C}$ 
and the operator $\tilde{\WW}=\WW\otimes \ketbra{\psi^*}$ on the two-copy space 
$\mathcal{H} \otimes \mathcal{H}.$ Then, $\WW$ is an SLOCC witness
for the class $S_{\ket{\psi}}$, if and only if the operator $\tilde{\WW}$ is an 
entanglement witness with respect 
to the split $(A_1 A_2|B_1 B_2|C_1C_2)$:
\begin{equation}
\label{5}
\braket{\eta|\WW|\eta} \geq 0 \quad \Leftrightarrow 
\quad \braket{\xi_{\rm sep}|\tilde{\WW}|\xi_{\rm sep}} \geq 0, 
\end{equation}
where $\ket{\xi_{\rm sep}}$ are product states within the two-copy system, that is they 
are of the form $\ket{\xi_{\rm sep}}=\ket{\alpha_{A_1 A_2}} \otimes \ket{\beta_{B_1 B_2}} 
\otimes \ket{\gamma_{C_1C_2}}$ and $\ket{\eta}\in S_{\ket{\psi}}$.
\end{theorem}

\begin{proof}
The ``only if'' part( ``$\Rightarrow$'') of the proof can be shown as 
follows:

One can always write the witness operator $\WW$ in its eigenbasis 
$\WW=\sum_n \kappa_n \ketbra{\alpha^{(n)}}$ and therefore 
\begin{equation}
\label{6}
\braket{\eta|\WW|\eta} 
= \sum_n \kappa_n |\braket{\psi|A^{\dagger}\otimes B^{\dagger}\otimes C^{\dagger}|\alpha^{(n)}}|^2 \geq 0.
\end{equation}
Moreover, it holds that
\begin{equation}
\label{8}
\braket{\psi|A^{\dagger}\otimes B^{\dagger}\otimes C^{\dagger}|\alpha^{(n)}} 
= \tr(A^{\dagger}\otimes B^{\dagger} \otimes C^{\dagger} \ket{\alpha^{(n)}}\bra{\psi}).
\end{equation}
We consider a single summand in Eq.~(\ref{6}) and use the
following representation of the SLOCC operations $A$, $B$, $C$ 
and the state $\ket{\psi}$. We write $A=\sum_{ij} A_{ij} \ket{i}\bra{j}$, 
$B=\sum_{i'j'} B_{i'j'} \ket{i'}\bra{j'}$, 
$C=\sum_{i''j''} C_{i''j''} \ket{i''}\bra{j''}$,
$\ket{\alpha^{(n)}}=\sum_{kk'k''}\alpha_{kk'k''}^{(n)}\ket{kk'k''}$ 
and $\ket{\psi}=\sum_{ll'l''}\psi_{ll'l''}\ket{ll'l''}$.
Then we have 
\begin{align}
\nonumber
\tr(A^{\dagger}\otimes B^{\dagger} & \otimes C^{\dagger}  \ket{\alpha^{(n)}}\bra{\psi})
=
\nonumber
\\
&=
\!\!\!
\sum_{i i' i'' j j' j''} 
\!\!\!
A^*_{ij} B^*_{i'j'} C^*_{i''j''}
\alpha_{ii'i''}^{(n)}\psi_{jj'j''}^*
\nonumber
\\ 
&\equiv 
\langle\!\braket{A_{12}\otimes B_{12}\otimes C_{12}|\alpha_1^{(n)},\psi_2^*}\!\rangle, \label{9}
\end{align}
where the indices $1$ and $2$ indicate now the copies of the system and we use 
ket-vectors like $\ket{Y_{12}}\!\rangle=\sum_{ij} Y_{ij} \ket{ij}$ on the 
two-copy Hilbert space of each particle $Y\in\{A,B,C\}$.
In the same way we obtain:
\begin{align}
\label{90}
\tr(\ket{\psi}\bra{\alpha^{(n)}} A \otimes B \otimes C) 
& \equiv \langle\!\braket{\alpha_1^{(n)},\psi_2^*| A_{12}\otimes B_{12}\otimes C_{12}}\!\rangle.
\end{align} 
Thus Eq.~(\ref{6}) can be written as
\begin{equation}
\label{10}
 \langle\!\bra{A_{12} \otimes B_{12} \otimes C_{12}} \WW_1 \otimes \ketbra{\psi^*}_2 \ket{ A_{12}\otimes B_{12} \otimes C_{12}}\!\rangle \geq 0.
 \end{equation}
 So far, the vectors $\ket{Y_{12}}\!\rangle$ with $Y\in\{A,B,C\}$ are not entirely arbitrary, 
 as the operators $A, B$ and $C$ are invertible. However, as any non-invertible matrix can be
 approximated arbitrarily well by invertible matrices and the expression under consideration
 is continuous, the positivity condition in Eq.~(\ref{10}) holds for any vectors 
 $\ket{Y_{12}}\!\rangle$. Let us finally note that it is straightforward to see that if  $\WW$ is 
 not positive semidefinite then $\tilde{\WW}$ is not positive semidefinite as well. This 
 completes the ``only if'' part of the proof. 
 
The "if" part of theorem (``$\Leftarrow$'') follows from the fact that
Eq.~(\ref{10}) for all $\ket{Y_{12}}\!\rangle$ implies Eq.~(\ref{6});
moreover, $\tilde{\WW}=\WW\otimes \ketbra{\psi^*}$ being not positive 
semidefinite implies that $\WW$ is not positive semidefinite.
\end{proof}

In order to start the discussion, we first note that statement of the theorem
clearly holds for any number of parties, the proof can directly be generalized.
Also, we note that the complex conjugation $\ket{\psi^*}$ is relevant, as
there are instances where $\ket{\psi^*}$ and $\ket{\psi}$ are not equivalent
under SLOCC \cite{connylang, kraus2010}.

Second, we compare the theorem with known results. The theorem presents 
a generalization of the main result from Ref.~\cite{hulpke} from the bipartite 
to the multipartite case. The SLOCC classes in the bipartite case are characterized 
by the Schmidt number and the Schmidt witnesses considered in Ref.~\cite{hulpke} 
are just the SLOCC witnesses for the bipartite case. A similar connection for
the special case of bipartite witnesses for Schmidt number one has also been 
discussed in Ref.~\cite{TensProdEW}. Furthermore, for the multipartite case, where 
the Schmidt number classification is a coarse graining of the SLOCC classification,
a connection between Schmidt witnesses and entanglement witnesses has been proved
in Ref.~\cite{sperling}. This connection, however, is not equivalent to ours, 
as the dimension of the enlarged space in Ref.~\cite{sperling} is in general
larger.

Third, Theorem \ref{Th.1} provides the possibility to consider the problem of 
maximizing the overlap of two states under SLOCC from a different perspective. 
That is, by solving the problem of finding the minimal value of $\lambda$, for 
which $\tilde{\WW}=(\lambda \eins - \ketbra{\varphi}) \otimes \ketbra{\psi^*}$ 
is an entanglement witness for full separability one can determine the value 
of the maximal overlap between $\ket{\varphi}$ and $\ket{\psi}$ under SLOCC 
operations. In order to provide a concrete application of Theorem \ref{Th.1}, 
we derived the maximal squared overlap between an $N$-qubit GHZ state
\begin{equation}
\ket{GHZ}= \frac{1}{\sqrt{2}} (\ket{00\cdots 0} + \ket{11\cdots 1}) 
\end{equation}
and the SLOCC class of the $N$-qubit W state
\begin{equation}
\ket{W}= \frac{1}{\sqrt{N}} (\ket{10\cdots 0} + \ket{01\cdots 0} + \dots + \ket{00\cdots 1}) 
\end{equation}
using the relation derived above 
in the Appendix. The resulting value is $3/4$ for $N=3$ 
(numerically already known from Ref.~\cite{3qubitmixed})  
and $1/2$ for $N\geq 4$ (for four-qubit states this value has 
been already found in Ref.~\cite{symmetricmixedSLOCC}). It should 
be noted that there is an asymmetry: While the SLOCC class of the 
three-qubit W state can approximate the GHZ state only to a certain 
degree, one can find arbitrarily close to the W state a state in the 
SLOCC orbit of the GHZ state \cite{3qubitmixed}.

Finally, our result reflects that the separability problem as well as the 
problem of deciding whether two tripartite states are within the same SLOCC 
class are both computationally highly non-trivial. In fact, they were shown 
to be NP-hard \cite{sepNP, gharibian, SLOCCNP}. 

In the following section we will discuss a relaxation of witness condition 
to be positive on all separable states. Instead one can consider the condition
that $\tilde \WW$ should  be positive on states having a positive partial 
transpose (PPT) for any bipartition.

\subsection{Using entanglement criteria for the witness 
construction}

In general, it can be very difficult to find an analytical solution for the 
minimal value of $\lambda$ such that the expectation value of 
$\tilde{\WW}=(\lambda \eins - \ketbra{\varphi}) \otimes \ketbra{\psi^*}$ 
is positive on all product states $\ket{\xi_{\rm sep}}$. To circumvent this 
problem, one can try to broaden the restrictions on the set of states on 
which $\tilde{\WW}$ is positive in a way that the new set naturally includes
the original set of separable states. 

One potential way to do that uses the the criterion of the positivity of
the partial transpose (PPT), as the set of separable states is a subset
of the states which are PPT \cite{peresppt}. More precisely, one can demand that 
$\tilde{\WW}$ is positive on the set of states which are PPT with 
respect to all subsystems in the considered bipartite splittings, 
i.e.,
\begin{align}
&\tr(\varrho_{A_{12}B_{12}C_{12}} \tilde{\WW})\geq 0 \nonumber \\
& \mbox{ for all}~\varrho_{A_{12}B_{12}C_{12}} \mbox{ with: } \varrho^{T_{Y_{12}}} \geq 0,~Y=\{A,B,C\}.
\label{11}
\end{align}
Although the set of PPT states is known to include PPT entangled states, this relaxation 
of the initial conditions offers an advantage, as we are able to formulate the problem 
of determining $\lambda$ as a semi-definite program (SDP)and as such provides a way for an 
exact result \cite{SDP}. For a given $\lambda$ one can consider the optimization problem
\begin{align}
   \mbox{minimize: }&\tr(\varrho \tilde{\WW})\nonumber \\
   \mbox{subject to: }& \varrho \geq 0,\nonumber\\
  & \varrho^{T_i} \geq 0 \mbox{ for } i=A,B,C, \nonumber \\
  & \tr(\varrho)=1.
 \label{12}
\end{align}
Such optimization problems can be solved with standard computer algebra systems. 
If the obtained value in Eq.~(\ref{12}) is non-negative, the initial operator 
$\WW = \lambda \eins - \ketbra{\varphi}$ was an SLOCC witness, so $\lambda$
is an upper bound on the maximal overlap.

To give an example, one may use this optimization for obtaining an upper bound 
on the overlap between the four-qubit cluster state and the SLOCC orbit of the 
four-qubit GHZ state or vice versa. In all the interesting examples, however, 
one obtains only the trivial bound $\lambda=1$. This finds
a natural explanation: If $\lambda$ is the exact maximal overlap, then the 
witness $\tilde \WW$ detects some entangled states which are PPT with respect 
to any bipartition. Consequently, relaxing the positivity on separable states 
to positivity on PPT state is a rather wasteful approximation in our case, and the 
resulting estimate on $\lambda$ is also wasteful. 

The key observation is that given two pure bipartite states, $\ket{\phi}$ and 
$\ket{\psi^*}$ in a $d_1 \times d_1$ and $d_2 \times d_2$ system, respectively, 
the total state 
\begin{align}
\sigma
= &\frac{1-p}{(d_1-1)(d_2-1)}
(\eins_1-\ket{\phi}_1\!\bra{\phi}) 
\otimes
(\eins_2-\ket{\psi^*}_2\!\bra{\psi^*})
\nonumber \\
&+p\ket{\phi}_1\!\bra{\phi}\otimes\ket{\psi^*}_2\!\bra{\psi^*},
\label{bestate}
\end{align}
as a state on a $d_1  d_2  \times d_1 d_2$-system is PPT, but typically entangled. 
This holds for nearly arbitrary choices for $\ket{\phi}$ and $\ket{\psi^*}$ and 
small values of $p$ \cite{pianistates}. Note that states of the form given in 
Eq.~(\ref{bestate}) lead to 
$\tr [(\lambda\eins_1-\ket{\phi}_1\!\bra{\phi})\otimes\ket{\psi^*}_2\!\bra{\psi^*})\sigma]
<0$ for any $\lambda<1$, so they are detected by the witness $\tilde \WW$. Hence, the 
relaxation to states that are PPT does, for general $\ket{\phi}$ and $\ket{\psi}$ not 
allow to determine possible non-trivial values of $\lambda$ for which $\tilde{\WW}$ 
is an entanglement witness. 

We mention that in Ref.~\cite{pianistates} operators of the form $(\lambda\eins-\ket{\phi}\bra{\phi})\otimes (\ket{\psi^*}\bra{\psi^*})$ 
with an appropriate choice of $\lambda$ have been shown to be bipartite 
entanglement witnesses for the case where the Schmidt rank of 
$\ket{\psi^*}$ is smaller than the Schmidt rank of $\ket{\phi}$ for the 
considered bipartite splitting. This can be easily understood using our 
result and the results of Ref.~\cite{hulpke}, as in this case $\ket{\phi}$ 
and $\ket{\psi}$ are in different bipartite SLOCC classes and $\ket{\phi}$ 
cannot be approximated arbitrarily close by a state in the SLOCC class 
of $\ket{\psi}$.

Finally, we add that considering other relaxations of the set of separable 
states can provide a way to estimate the maximal SLOCC overlap using an
SDP. Here, other positive maps besides the transposition,
such as the Choi map \cite{hororeview}, or the SDP approach 
of Ref.~\cite{doherty} seems feasible.


\section{SLOCC overlaps for $2\times 3\times 3$ systems}

Systems consisting of one qubit, one qutrit and one system of 
arbitrary dimension mark the last cases, which still have a 
finite number of SLOCC classes \cite{SLOCCfinite}, and for general 
systems the number of SLOCC classes is infinite \cite{SLOCC2222}. For 
one qubit and two qutrits there are 17 different classes with 
12 of these being truly tripartite entangled and six of them 
containing entangled states with maximal Schmidt rank across
the bipartitions \cite{SLOCCfinite, ChMi10}. Finding the maximal overlap 
of the representative states of the different classes not only 
indicates towards an hierarchy among them, but, as shown in Section 
III, gives insight in the entanglement properties of states in an 
enlarged two-copy system. In fact, one can then construct 
entanglement witnessrs, $\tilde{\WW}$ which detect entanglement 
within states of dimension $ 4 \times 6 \times 6 $.  Thus, for 
all pairs of representatives and SLOCC classes where $\lambda < 1$ 
one can construct a specific $\tilde{\WW}$ which, as discussed
above, typically also detects PPT entangled states.

The unnormalized representative states of the fully entangled
SLOCC classes within a $2 \times 3\times 3$ system are 
\cite{SLOCCfinite}:
\begin{align}
 \ket{\psi_6}&=\ket{000}+\ket{111}, \nonumber \\
 \ket{\psi_7}&=\ket{000}+\ket{011}+\ket{101}, \nonumber \\
 \ket{\psi_8}&=\ket{000}+\ket{011}+\ket{102}, \nonumber \\
 \ket{\psi_9}&=\ket{000}+\ket{011}+\ket{120}, \nonumber\\
 \ket{\psi_{10}}&=\ket{000}+\ket{011}+\ket{122}, \nonumber \\
 \ket{\psi_{11}}&=\ket{000}+\ket{011}+\ket{101}+\ket{112}, \nonumber \\
 \ket{\psi_{12}}&= \ket{000}+\ket{011}+\ket{110}+\ket{121}, \nonumber \\
 \ket{\psi_{13}}&=\ket{000}+\ket{011}+\ket{102}+\ket{120}, \nonumber \\
 \ket{\psi_{14}}&=\ket{000}+\ket{011}+\ket{112}+\ket{120}, \nonumber \\
 \ket{\psi_{15}}&=\ket{000}+\ket{011}+\ket{100}+\ket{122}, \nonumber \\
 \ket{\psi_{16}}&=\ket{000}+\ket{011}+\ket{022}+\ket{101}, \nonumber \\
 \ket{\psi_{17}}&=\ket{000}+\ket{011}+\ket{022}+\ket{101}+\ket{112}.
\end{align}
One can compute the  overlap between one of these states and the 
SLOCC orbit of another state via direct optimization. As for the 
GHZ class and the W state, it can happen that one class can approximate
one state arbitrarily well, so we set the overlap to one, if the numerical
obtained value approximates this with a  numerical precision of $10^{-12}$. 
Note that an exact value of one is impossible, as the SLOCC classes are
proven to be different.

\begin{table*}[t!!]
\begin{small}
 \begin{tabular}{| c| c | c | c | c | c | c | c | c | c | c | c | c | }
    \hline 
   $\cdot\backslash\cdot$ &$\!{\ket{\psi_6}}\!$ & $\!{\ket{\psi_7}}\!$ & $\!{\ket{\psi_8}}\!$ & $\!{\ket{\psi_9}}\!$& $\!{\ket{\psi_{10}}}\!$ & $\!{\ket{\psi_{11}}}\!$ & $\!{\ket{\psi_{12}}}\!$ 
   & $\!{\ket{\psi_{13}}}\!$& $\!{\ket{\psi_{14}}}\!$ 
    & $\!{\ket{\psi_{15}}}\!$ & $\!{\ket{\psi_{16}}}\!$ & 
    $\!{\ket{\psi_{17}}}\!$ 
    \\ \hline \hline
      ${\ket{\psi_6}}$  & $\diamond$ & $1$ & $ 2 / 3$ & $2 / 3$&  $2 / 3$& $3 / 4$ & $3 / 4$ & $3 / 4$ & $\!0.5625\!$ 
     &$ 3 / 4 $ & $ 3 / 4$ & $0.65$ \\ \hline
    ${\ket{\psi_7}}$  & $3 / 4$ & $\diamond$ & $ 2 / 3$ & $2 / 3$&  $2 / 3$& $3 / 4$ & $3 / 4$ & $\!0.5433\!$& $\!0.5625\!$ 
    & $0.7$ & $3 / 4 $ & $\!0.6129\!$ \\ \hline
    ${\ket{\psi_8}}$  & $1$ & $1$ & $\diamond$ & $ 2 / 3 $& $2 / 3$ & $0.875$ & $3 / 4$ & $3 / 4$& $3 / 4$ 
    & $3 / 4$ & $ 3 / 4$ & $\!0.7252\!$ \\ \hline
    ${\ket{\psi_9}}$ & $1$ & $1$ & $2 / 3$ & $\diamond$& $2 / 3$ & $3 / 4$ & $0.875$ & $3 / 4$& $3 / 4$ 
    & $3 / 4$ & $ 3 / 4$ & $\!0.7252\!$  \\ \hline 
    $ \!{\ket{\psi_{10}}}\!$ & $1$ & $1$ & $1$ & $1$& $\diamond$ & $0.875$ & $0.875$ & $\!0.8333\!$& $3 / 4$ 
    & $\!0.9045\!$ & $1$ & $\!0.7955\!$ \\ \hline
    $\! {\ket{\psi_{11}}}\!$ & $1$ & $1$ & $1$ & $2 / 3$& $2 / 3$ & $\diamond$ & $3 / 4$ & $3 / 4$& $3 / 4$ 
    & $3 / 4$ & $3 / 4$ & $0.8$ \\ \hline 
    $\! {\ket{\psi_{12}}}\!$ & $1$ & $1$ & $ 2 / 3$ & $1$& $2 / 3$ & $3 / 4$ & $\diamond$ & $3 / 4$& $3 / 4$ 
    & $3 / 4$ & $3 / 4$ & $0.8$ \\ \hline
    $\! {\ket{\psi_{13}}}\!$ & $1$ & $1$ & $1$ & $1$& $\!0.8333\!$ & $1$ & $1$ & $\diamond$& $1$ 
    & $0.95$ & $1$ & $1$ \\ \hline 
   $\! {\ket{\psi_{14}}}\!$& $1$ & $1$ & $1$ & $1$& $2 / 3$ & $0.875$ & $0.875$ & $\!0.8125\!$& $\diamond$ 
    & $3 / 4$ & $3 / 4$ & $0.8$ \\ \hline
   $\!{\ket{\psi_{15}}}\!$ & $1$ & $1$ & $1$ & $1$& $1$ & $1$ & $1$ & $1$& $1$ 
    & $\diamond$ & $1$ & $1$ \\ \hline 
    $\! {\ket{\psi_{16}}}\!$ & $1$ & $1$ & $1$ & $1$& $\!0.7357\!$ & $0.875$ & $0.875$ & $\!0.7706\!$& $3 / 4$ 
    & $3 / 4$ & $\diamond$ & $0.8$ \\ \hline
    $\!{\ket{\psi_{17}}}\!$& $1$ & $1$ & $1$ & $1$& $0.795$ & $1$ & $1$ & $\!0.8958\!$& $1$ 
    & $0.875$ & $1$ & $\diamond$ \\ \hline
    \end{tabular}
 \end{small}
  \caption{This table shows the numerical values for the maximal 
  squared overlap between $\ket{\psi_i}$ (column) and the SLOCC
     orbit of $\ket{\psi_j}$ (row). See text for further details.
     \label{SLOCCoverlap233}}
 \end{table*}

The values of the numerical maximization of the  SLOCC overlap for the 
different SLOCC classes with respect to the representative states from 
above is given in Table \ref{SLOCCoverlap233}. They should be interpreted as follows: 
For the overlaps between $\ket{\psi_6}$ and $\ket{\psi_7}$ two
different values are given. The value $\lambda=1$ means that
the SLOCC orbit of state $\ket{\psi_6}$ approximates $\ket{\psi_7}$
arbitrarily well. The value $\lambda=3/4$ means that the SLOCC orbit
of $\ket{\psi_7}$ cannot approximate $\ket{\psi_6}$ so well, only
an overlap of $\lambda=3/4$ can be reached. This implies that
$\WW= 3/4 \times \eins - \ketbra{\psi_6}$ is an SLOCC witness, 
discriminating $\ket{\psi_6}$ from the SLOCC orbit of $\ket{\psi_7}$. 
Note that $\ket{\psi_6}$ and $\ket{\psi_7}$ are essentially the
three-qubit GHZ- and W states encountered above.

This also has consequences for the classification of mixed states, see
Fig.~\ref{SLOCCpic233}. For a mixed state, one may ask whether it can be
written as a convex combination of pure states within some SLOCC class.
If a state can be written as such a convex combination of states from
the orbit of $\ket{\psi_7}$, it can also be written with states from
the orbit of $\ket{\psi_6}$, as the latter can approximate the former
arbitrarily well. Consequently, there is an inclusion relation for the mixed
states, as depicted in Fig.~\ref{SLOCCpic233}.

 \begin{figure}[t]
 \begin{center}
  \includegraphics[width=0.7\columnwidth]{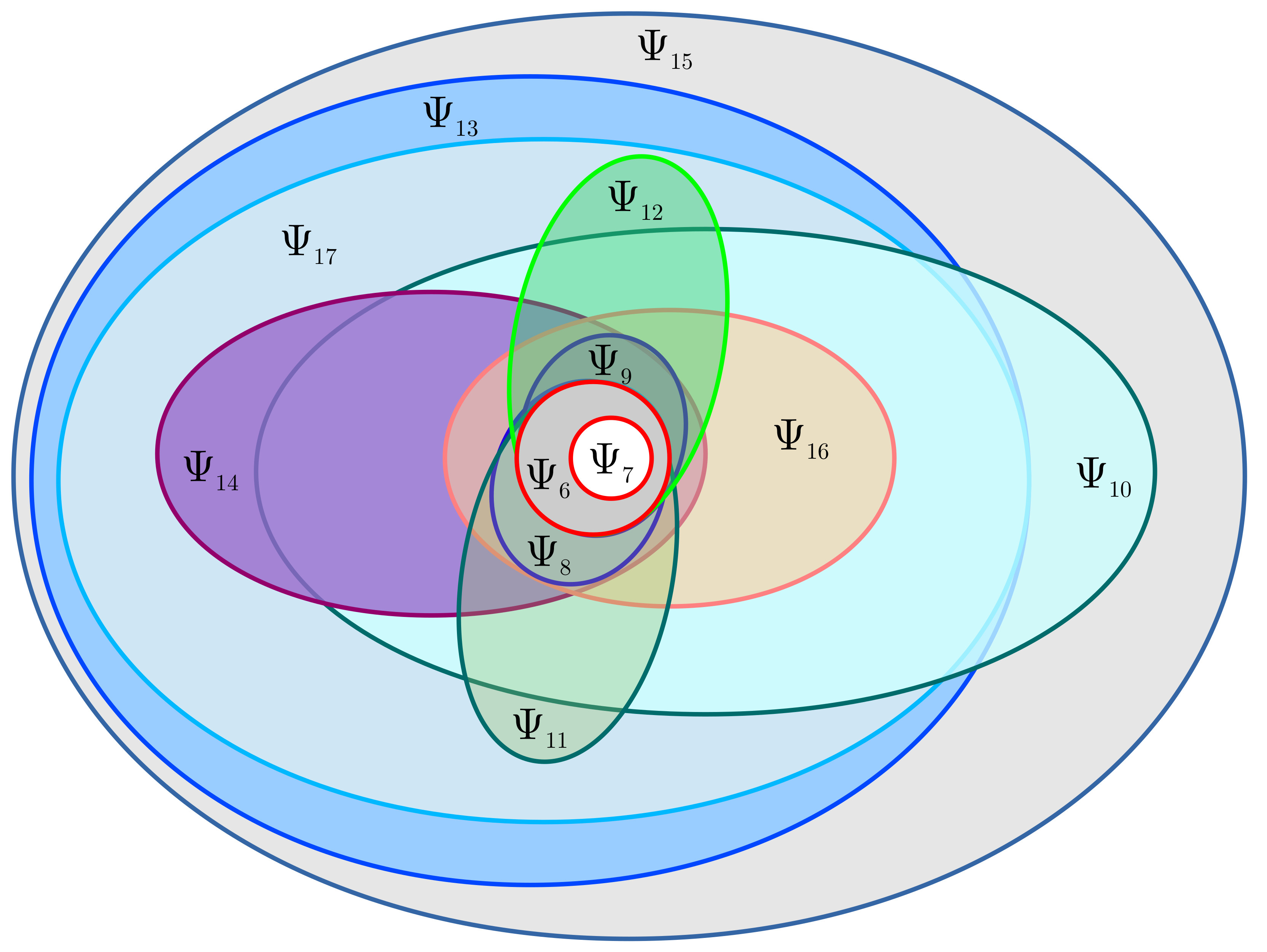} 
  \end{center}
     \caption{Hierarchic structure of SLOCC classes for mixed states within a 
     qubit-qutrit-qutrit system. If one pure state orbit of class $\ket{\psi_i}$
     can be approximated by another SLOCC orbit $\ket{\psi_j}$ arbitrary well, 
     the corresponding mixed states in class $i$ are included in the mixed states
     in class $j$. As can be seen from Table~\ref{SLOCCoverlap233}, $\ket{\psi_{15}}$ 
     is the most powerful class in the sense that any other state $\ket{\psi_i}$ can 
     be reached from $\ket{\psi_{15}}$ via SLOCC operations with arbitrary high accuracy.
     \label{SLOCCpic233}} 
\end{figure}

\section{Conclusions}

For arbitrary numbers of parties and local dimensions we showed 
a one-to-one correspondence between an operator $\WW$ able to 
distinguish between different SLOCC classes of a system and 
another operator $\tilde{\WW}$ that detects entanglement
in a two-copy system. This correspondence thereby enables us to 
directly transfer a solution for one problem to the other. 
Though the relaxation to PPT states in order to construct the 
entanglement witness did not prove to be helpful for reasons 
stated in Section III, it very well might be that other possible 
relaxations on the set of separable states will give more insight 
and a good approximation for an upper bound on the maximal overlap. 
As an concrete application of the presented relation we derived the 
maximal overlap between the $N$-qubit GHZ state and states within the 
$N$-qubit W class. The calculations in Section IV for the 
qubit-qutrit-qutrit system do not only indicate a hierarchy among
the SLOCC classes but also provides us with the option to construct 
a whole set of entanglement witnesses for the doubled system of 
dimensions $4 \times 6\times 6$.

\section{Acknowledgements}
We thank Florian Hulpke and Jan Sperling for helpful discussions. 
This work has been supported by the ERC (Consolidator Grant 683107/TempoQ), 
the DFG, and the Austrian Science Fund (FWF): J 4258-N27.

\section{Appendix: Maximal squared overlap between the GHZ-state and states in the W-class}

In this Appendix we will provide an example of how the relation among 
SLOCC witnesses and entanglement witnesses can be employed and compute the maximal squared 
overlap between the GHZ-state of $N$-qubits,
$\ket{GHZ_N}=1/\sqrt{2}(\ket{00\ldots 0}+\ket{11\ldots 1}$, 
and a normalized $N$-qubit state in the W-class
[with representative $\ket{W_N}=1/\sqrt{N}(\ket{10\ldots 0}+\ket{010\ldots 0}+\ldots +\ket{0\ldots 010}+\ket{0\ldots 01})$]. 
We show that for $N=3$ the maximal squared overlap is given by $\frac{3}{4}$ (see also \cite{3qubitmixed}), 
whereas for $N\geq 4$ it is given by $\frac{1}{2}$. For 4-qubit states this value has been already found in \cite{symmetricmixedSLOCC}.

In order to do so we consider  
\begin{equation}
\tilde{\WW}_{N}=(\lambda_N\eins-\ket{GHZ_N}\bra{GHZ_N})\otimes\ket{W_N}\bra{W_N}
\end{equation} and show that it is an 
entanglement witness (for $2N$-qubit states) with respect to the splitting $(A_1A_2|B_1B_2|C_1C_2|\ldots |Z_1Z_2)$ iff 
$1>\lambda_3\geq \frac{3}{4}\equiv \lambda_3^C$ and $1>\lambda_N\geq\frac{1}{2}\equiv \lambda_N^C$ for $N\geq 4$.
Using Theorem \ref{Th.1} this implies that $\bra{\Psi_W^N}(\lambda_N\eins-\ket{GHZ_N}\bra{GHZ_N})\ket{\Psi_W^N}\geq 0$, where $\ket{\Psi_W^N}$ 
denotes a normalized state in the $N$-qubit W-class, iff $1>\lambda_N\geq\lambda_N^C$. 
Recall that $\bra{\Psi_W^N}(\lambda_N\eins-\ket{GHZ_N}\bra{GHZ_N})\ket{\Psi_W^N}\geq 0$ is
equivalent to $\lambda_N\geq |\braket{GHZ_N|\Psi_W^N}|^2$ and therefore the maximal squared overlap is given by $\lambda_N^C$. 

Before considering the problem of finding the range of $\lambda_N$ for which 
$\tilde{\WW}_{N}$ is an entanglement witness let us first present a parametrization of states 
in the W-class that will be convenient for our purpose and then relate it to 
the parametrization of product states that have to be considered.
It is well known that any state in the W-class can be written as 
$\bigotimes_i U_i(x_0\ket{00\ldots 0}+x_1\ket{10\ldots 0}+x_2\ket{010\ldots 0}+\ldots +x_{N-1}\ket{0\ldots 010}+x_N\ket{0\ldots 01})$ 
with $x_0\geq 0$, $x_i> 0$ for $i\in\{1, \ldots, N\}$ and $U_i$ unitary
\cite{SLOCC222}. 
Note that we do not impose that the states are normalized. 
Equivalently, one can write it as 
$U_1D_1\otimes U_2D_2\otimes \ldots \otimes U_{N-2}D_{N-2}\otimes U_{N-1}g_{N-1}\otimes U_ND_{N} \ket{W_N}$ where
$D_i=\mbox {diag }(1, \tilde{x}_i)$ with $\tilde{x}_i=x_i/x_N >0$ and
\begin{equation}
g_{N-1}=\left(\begin{array}{cc}
x_{N}&x_0\\ 
0&x_{N-1}\\ 
\end{array} \right).
\end{equation}

For the local unitaries on the qubits we will use the parametrization $U_i=U_{ph}(\gamma_i)X(\alpha_i)U_{ph}(\beta_i)$ with $X(\delta) =e^{i \delta X}$, $U_{ph}(\delta)=\mbox {diag }(1, e^{i\delta})$ and $\alpha_i, \beta_i, \gamma_i \in \mathbb{R}$. In order to simplify our argumentation we will use the symmetry that $\bigotimes_i U_{ph}(\delta)\ket{W_n}=e^{i\delta}\ket{W_n}$ and choose $\beta_N=0$, $\beta_i=\beta_i-\beta_N$ for $i\in\{1,\ldots,N-2\}$ and $x_j=x_je^{-i\beta_N}$ for $j=0,N-1$. Furthermore, using 
for the GHZ state the symmetry that $U_{ph}(\delta_1)\otimes U_{ph}(\delta_2)\otimes\ldots \otimes U_{ph}(\delta_{N-2})\otimes U_{ph}(-\sum_{i\in I_0}\delta_i)\otimes U_{ph}(\delta_N)\ket{GHZ_N}=\ket{GHZ_N}$ where here and in the following $I_0=\{1,2,\ldots,N-2,N\}$ one can easily see that when computing the maximal SLOCC overlap between the GHZ state and a 
W class state one can equivalently choose $\gamma_i=0$ for $i\in I_0$ and $\gamma_{N-1}=\sum_{i=1}^N\gamma_i$. 

We will now make use of the fact that 
$\bra{\eta}(\lambda_N\eins-\ket{GHZ_N}\bra{GHZ_N})\ket{\eta}\geq 0$ for 
$\ket{\eta}=A\otimes B\otimes \ldots \otimes Z\ket{W_N}$  if and only if 
$\bra{\xi_{SEP}}[\lambda_N\eins-(\ket{GHZ_N}\bra{GHZ_N})_1]\otimes(\ket{W_N}\bra{W_N})_2\ket{\xi_{SEP}}
\geq 0$ for 
$\ket{\xi_{SEP}}= \ket{A_{12}}\otimes \ket{B_{12}} \otimes \ldots \ket{Z_{12}} $ 
with  
$\ket{\Gamma_{12}}=(\Gamma_1\otimes\eins_2)\ket{\Phi^+}$, $\Gamma\in\{A,B, \ldots Z\}$ 
and $\ket{\Phi^+}=\sum_{i=0}^{1}\ket{ii}$ (this follows from the proof of Theorem \ref{Th.1}). 
As any  state in the W class can be parametrized as explained above we only have to consider 
product states of the form  $\ket{\xi_{SEP}}= \otimes_{i=1}^N\ket{\phi_i}$ with 
$\ket{\phi_i}=(U_iD_i\otimes  \eins )\ket{\Phi^+}=(U_i\otimes  \eins )(\ket{00}+\tilde{x}_i\ket{11})$ 
for $i\in I_0$ and $\ket{\phi_{N-1}}=(U_{N-1}g_{N-1}\otimes  \eins )\ket{\Phi^+}$. 
As before the expectation value of $\tilde{\WW}_{N}$ for states with some separable $\ket{\phi_i}$ can 
be approximated arbitrarily close by the expectation value for a state $\ket{\xi_{SEP}}$ for which all $\ket{\phi_i}$ are entangled. Note that 
 $\bra{\xi_{SEP}}\tilde{\WW}_{N}\ket{\xi_{SEP}}\geq 0$  for all $\ket{\xi_{SEP}}$ as defined above iff the operator $\tilde{w}_N\equiv\bra{\zeta_{SEP}}\tilde{\WW}_{N}\ket{\zeta_{SEP}}\geq 0$ is positive 
 semidefinite for all $\ket{\zeta_{SEP}}=\otimes_{i\in I_0} \ket{\phi_i}$ with $\ket{\phi_i}$ as defined above. This is due to the fact that the parameters of $\ket{\zeta_{SEP}}$ and $\ket{\phi_{N-1}}$ can be chosen independently and $\ket{\phi_{N-1}}$ is an arbitrary state.
 
One obtains for the respective terms of $\tilde{w}_N$ that
\begin{equation}
\bra{\zeta_{SEP}}[\eins_1\otimes(\ket{W_N}\bra{W_N} )_2]\ket{\zeta_{SEP}}=\frac{1}{N}\eins_{\Gamma_{1}}\otimes[\eins_{\Gamma_2}+\sum_{i=1}^{N-2}\tilde{x}_i^2 (\ket{0}\bra{0})_{\Gamma_2}],
\end{equation}
where $\Gamma$ refers to party $N-1$. The other term can be written as $\bra{\zeta_{SEP}}(\ket{GHZ_N}\bra{GHZ_N})_{1}\otimes(\ket{W_N}\bra{W_N} )_{2}\ket{\zeta_{SEP}}=(\ket{\varphi}\bra{\varphi})_{\Gamma_1\Gamma_2}$ 
with 
\begin{align} 
\label{phi}
\ket{\varphi}_{\Gamma_1\Gamma_2}=
&\frac{1}{\sqrt{2N}}\{[\sum_{j\in I_0}(-i\sin (\alpha_j)\tilde{x}_je^{-i\beta_j}\prod_{k\in I_0\backslash\{j\}}\cos (\alpha_k))\ket{0}_{\Gamma_1}
\nonumber
\\
\nonumber
&+\sum_{j\in I_0}(\cos (\alpha_j)\tilde{x}_je^{-i\beta_j}\prod_{k\in I_0\backslash\{j\}}(-i\sin (\alpha_k)))\ket{1}_{\Gamma_1}]\otimes \ket{0}_{\Gamma_2}
\nonumber \\ \nonumber
&+[\prod_{j\in I_0}\cos (\alpha_j)\ket{0}_{\Gamma_1}+\prod_{j\in I_0}(-i \sin (\alpha_j))\ket{1}_{\Gamma_1})]\otimes\ket{1}_{\Gamma_2}\}\\
\equiv& \ket{\varphi_0}_{\Gamma_1}\ket{0}_{\Gamma_2}+ \ket{\varphi_1}_{\Gamma_1}\ket{1}_{\Gamma_2}.
\end{align}
Hence, we have that
$\tilde{w}_N=\frac{\lambda_N}{N}\eins_{\Gamma_{1}}\otimes[\eins_{\Gamma_2}+\sum_{i=1}^{N-2}\tilde{x}_i^2 (\ket{0}\bra{0})_{\Gamma_2}]-
(\ket{\varphi}\bra{\varphi})_{\Gamma_1\Gamma_2}$.
Defining $\mu=||\varphi_0||$ and $\nu=||\varphi_1||$ we can
write $\ket{\varphi}=\mu \ket{\Phi_0}_{\Gamma_1}\ket{0}_{\Gamma_2}+ \nu \ket{\Phi_1}_{\Gamma_1}\ket{1}_{\Gamma_2}$ where  $||\Phi_i||=1$.
We construct now the following orthonormal basis: 
\begin{align} 
&\ket{\Psi_0}=\frac{\mu}{\sqrt{\mu^2+\nu^2}} \ket{\Phi_0}_{\Gamma_1}\ket{0}_{\Gamma_2}+ \frac{\nu}{\sqrt{\mu^2+\nu^2}}
\ket{\Phi_1}_{\Gamma_1}\ket{1}_{\Gamma_2}
,
\\ 
&\ket{\Psi_1}=\frac{\nu}{\sqrt{\mu^2+\nu^2}} \ket{\Phi_0}_{\Gamma_1}\ket{0}_{\Gamma_2}- 
\frac{\mu}{\sqrt{\mu^2+\nu^2}} \ket{\Phi_1}_{\Gamma_1}\ket{1}_{\Gamma_2},
\\ &\ket{\Psi_2}=\ket{\Phi_0^{\perp}}_{\Gamma_1}\ket{0}_{\Gamma_2},
\\ 
&\ket{\Psi_3}=\ket{\Phi_1^{\perp}}_{\Gamma_1}\ket{1}_{\Gamma_2},
\end{align}
where $\braket{\Phi_i|\Phi_i^{\perp}}=0$ for $i\in\{0,1\}$. 

It can be easily seen that 
$ \tilde{w}_N=\sum_{i,j=0}^1\Lambda_{ij}\ket{\Psi_i}\bra{\Psi_j}+\frac{\lambda_N}{N} (1+\sum_{i=1}^{N-2}\tilde{x}_i^2)\ket{\Psi_2}
\bra{\Psi_2}+\frac{\lambda_N}{N} \ket{\Psi_3}\bra{\Psi_3}$ with 
\begin{equation}
\Lambda=\left(\begin{array}{cc}
\frac{\lambda_N}{N}(1+ \sum_{i=1}^{N-2}\tilde{x}_i^2\frac{\mu^2}{\mu^2+\nu^2})-(\mu^2+\nu^2)&\sum_{i=1}^{N-2}\tilde{x}_i^2\frac{\lambda_N \mu\nu}{N (\mu^2+\nu^2)} \\ 
 \sum_{i=1}^{N-2}\tilde{x}_i^2 \frac{\lambda_N\mu\nu}{N (\mu^2+\nu^2)} &\frac{\lambda_N}{N}(1+ \sum_{i=1}^{N-2}\tilde{x}_i^2\frac{\nu^2}{\mu^2+\nu^2})\\ 
\end{array} \right).
\end{equation}
Note that as we consider the case $\lambda_N>0$ (otherwise $\tilde{\WW}_{N}\leq 0$ which implies that it cannot be an entanglement witness) and as $\tilde{x}_i\in R$
we have that $ \tilde{w}_N\geq 0$ iff $\Lambda \geq 0$. In order to determine for which values of $\lambda_N$ the matrix $\Lambda$ is a positive semidefinite matrix we impose that $\tr(\Lambda)\geq 0$ and $\det (\Lambda)\geq 0$.  It can be easily seen that $\det (\Lambda)\geq 0$ implies $\tr(\Lambda)\geq 0$ and one straightforwardly obtains that $\Lambda \geq 0$ iff $\frac{\lambda_N}{N}\geq \frac{\mu^2}{ \sum_{i\in I_0}\tilde{x}_i^2}+\nu^2$. Hence, the minimal $\lambda_N$ for which $\tilde{\WW}_{N}$ is an entanglement witness is given by 
\begin{equation}
\label{lambdamin}
\lambda_{N}^C=\underset{\tilde{x}_i,\alpha_i,\beta_i\in\mathbb{R}}{\sup}N(\frac{\mu^2}{\sum_{i\in I_0}\tilde{x}_i^2}+\nu^2).
\end{equation}
One can easily derive from Eq.~(\ref{phi}) that 
\bea
\mu^2&=&\frac{1}{2N}\Big[|\sum_{j\in I_0}\sin (\alpha_j)\tilde{x}_je^{-i\beta_j}\prod_{k\in I_0\backslash\{j\}}\cos (\alpha_k))|^2+
\nonumber \\
&&+
|\sum_{j\in I_0}(\cos (\alpha_j)\tilde{x}_je^{-i\beta_j}\prod_{k\in I_0\backslash\{j\}}\sin (\alpha_k))|^2\Big]
\eea 
and
\be
\nu^2=\frac{1}{2N}[\prod_{j\in I_0}\cos^2 (\alpha_j)+\prod_{j\in I_0}\sin^2 (\alpha_j)].
\ee
Note that as $|\sum_i a_i|\leq \sum_i |a_i|$ for any complex numbers $a_i$ (and as any possible pair of values of $|\sin (\delta)|$ and $|\cos (\delta)|$ is attained for $\delta\in[0,\pi/2]$ and $\sin (\delta)\geq 0$ and $\cos (\delta)\geq 0$ for this parameter range) one obtains that the supremum in Eq. (\ref{lambdamin}) is attained for $\beta_i=0$ and $\alpha_i\in[0,\pi/2]$. 

We will in the following distinguish between $N=3$ and $N\geq 4$ and first discuss the case $N=3$. 
Inserting the corresponding expressions for $\mu^2$ and $\nu^2$ in Eq.~(\ref{lambdamin}) and using $\beta_1=\beta_3=0$ one straightforwardly obtains that 
\be
\lambda_{3}^C=\underset{x,\alpha_1,\alpha_3\in\mathbb{R}}{\sup}\,\,\,\frac{1}{2}[1+\frac{x}{1+x^2}\sin (2\alpha_1)\sin (2\alpha_3)].
\ee
It is easy to see that therefore the supremum is obtained for $\alpha_1=\alpha_3=\pi/4$ and $x=1$ which implies that $\lambda_{3}^C=\frac{3}{4}$. Hence, if $\lambda_3$ is larger than $\frac{3}{4}$ the operator $\tilde{w}_3$ is positive semidefinite. However, it should be noted that $\tilde{\WW}_{3}$ is only an entanglement witness if  $\lambda_3<1$ as for $\lambda_3\geq 1$ the operator $\tilde{\WW}_{3}$ is positive semidefinite and there exists no state that is detected. A state that attains the maximum overlap of $3/4$ is given by $1/\sqrt{3}(\ket{+++}+\ket{--+}+\ket{+--})$ with $\ket{\pm}=1/\sqrt{2}(\ket{0}\pm\ket{1})$. Using $\lambda_3=3/4$, $\beta_1=\beta_3=0$, $x=1$ and  $\alpha_1=\alpha_3=\pi/4$ the remaining 
parameters for a state in the W class that attains the maximum can be obtained by calculating the eigenvector of $\tilde{w}_3$ for the eigenvalue $0$. Note that in order to obtain the state presented here symmetries of the GHZ and W state have been used.

We will proceed with $N\geq 4$ and will use that the supremum is attained for $\beta_i=0$. 
Note that then $\frac{\mu^2}{\sum_{i\in I_0}\tilde{x}_i^2}$ can be equivalently written as
\bea
(\vec{v}_0\cdot\vec{v}_1)^2+(\vec{v}_0\cdot\vec{v}_2)^2,
\eea
where 
\begin{align}
  \vec{v}_0&=\frac{1}{\sqrt{\sum_{i\in I_0}\tilde{x}_i^2}}(\tilde{x}_1,\tilde{x}_2,\ldots,\tilde{x}_{n-2},\tilde{x}_n)\\
  \vec{v}_1&=(y_1,\ldots, y_{N-2},y_N)\,\,\text{ with: } y_j=\frac{1}{\sqrt{2N}}\sin (\alpha_j)\prod_{k\in I_0\backslash\{j\}}\cos (\alpha_k)\\
   \vec{v}_2&=(z_1,\ldots, z_{N-2},z_N)\,\,\text{ with:}  z_j=\frac{1}{\sqrt{2N}}\cos (\alpha_j)\prod_{k\in I_0\backslash\{j\}}\sin (\alpha_k).
\end{align}
Hence, one obtains
\begin{equation}
\lambda_{N}^C=\underset{\tilde{x}_i,\alpha_i\in\mathbb{R}}{\sup}N[(\vec{v}_0\cdot\vec{v}_1)^2+(\vec{v}_0\cdot\vec{v}_2)^2+
\nu^2]\leq \underset{\alpha_i\in\mathbb{R}}{\sup}N[|\vec{v}_1]^2+|\vec{v}_2|^2+\nu^2]
\end{equation}
as $\vec{v}_0$ is a normalized vector. Inserting the expressions 
for $\vec{v}_1, \vec{v}_2$ and $\nu$ we have that 
\begin{align}
  \lambda_{N}^C\leq& \underset{\alpha_i\in\mathbb{R}}{\sup}\frac{1}{2}(\sum_{j\in I_0}\cos^2 (\alpha_j)\prod_{k\in I_0\backslash\{j\}}\sin^2 (\alpha_k)+\sum_{j\in I_0}\sin^2 (\alpha_j)\prod_{k\in I_0\backslash\{j\}}\cos^2 (\alpha_k)\\\nonumber
&+\prod_{j\in I_0}\cos^2 (\alpha_j)+\prod_{j\in I_0}\sin^2 (\alpha_j))\\\nonumber
=&\underset{\alpha_i\in\mathbb{R}}{\sup}\frac{1}{2}(\sum_{j\in I_0\backslash\{N\}}\cos^2 (\alpha_j)
\prod_{k\in I_0\backslash\{j\}}\sin^2 (\alpha_k)+\sum_{j\in I_0\backslash\{N\}}\sin^2 (\alpha_j)\prod_{k\in I_0\backslash\{j\}}\cos^2 (\alpha_k)\\\nonumber
&+\prod_{j\in I_0\backslash\{N\}}\cos^2 (\alpha_j)+\prod_{j\in I_0\backslash\{N\}}\sin^2 (\alpha_j))\\\nonumber
\leq&\underset{\alpha_i\in\mathbb{R}}{\sup}\frac{1}{2}(\sum_{j\in I_0\backslash\{N\}}\cos^2 
(\alpha_j)\prod_{k\in I_0\backslash\{j,N\}}\sin^2 (\alpha_k)+\sum_{j\in I_0\backslash\{N\}}\sin^2
(\alpha_j)\prod_{k\in I_0\backslash\{j,N\}}\cos^2 (\alpha_k)\\\nonumber
&+\prod_{j\in I_0\backslash\{N\}}\cos^2 (\alpha_j)+\prod_{j\in I_0\backslash\{N\}}\sin^2 (\alpha_j))\\\nonumber
&\leq \underset{\alpha_i\in\mathbb{R}}{\sup}\frac{1}{2}(\sum_{j\in \{1,2,3\}}\cos^2 (\alpha_j)
\prod_{k\in\{1,2,3\},k\neq j}\sin^2 (\alpha_k)+\sum_{j\in\{1,2,3\}}\sin^2 (\alpha_j)\prod_{k\in \{1,2,3\},k\neq j}\cos^2 (\alpha_k)\\\nonumber
+&\prod_{j\in \{1,2,3\}}\cos^2 (\alpha_j)+\prod_{j\in \{1,2,3\}}\sin^2 (\alpha_j))\\\nonumber
&=\frac{1}{2}.
 \end{align}
 Note that for the second inequality we used that $0\leq\cos^2(\alpha_i)\leq 1$ and $0\leq\sin^2(\alpha_i)\leq 1$ and then repeatedly 
applied the same argumentation.  Note further that the upper bound obtained in the last line is equal to $1/2$ independent of the value of the parameters $\alpha_i$ for $i\in\{1,2,3\}$. As the state $\ket{00\ldots 0}$ which can be approximated arbitrarily close by a state in the W
class has a squared overlap with the GHZ
state of $1/2$ we also have that $\lambda_{N}^C\geq 1/2$. Hence, one obtains $\lambda_{N}^C=1/2$ for $N\geq 4$. Note that this is also the maximal squared overlap between the GHZ
state and an arbitrary separable state.



\section*{References}
\begin{thebibliography}{99}

\bibitem{hororeview}
R. Horodecki, P. Horodecki, M. Horodecki,
and K. Horodecki, Rev. Mod. Phys. {\bf 81}, 865 (2009).

\bibitem{gtreview}
O. G\"uhne and G. T\'oth, 
Phys. Rep. {\bf 474}, 1 (2009).

\bibitem{Nielsen} 
M. A. Nielsen, Phys. Rev. Lett. \textbf{83}, 436 (1999).

\bibitem{kleinmann}
M. Kleinmann, H. Kampermann, and D. Bru{\ss},
Phys. Rev. A {\bf 84}, 042326 (2011). 

\bibitem{ChLe12} E. Chitambar, D. Leung, L. Mancinska, M. Ozols, A. 
Winter, Commun. Math. Phys., {\bf 328}, 1, 303 (2014). 

\bibitem{MES} J. I. de Vicente, C. Spee, and  B. Kraus, Phys. Rev. Lett. \textbf{111}, 110502 (2013).


\bibitem{connylang}
C. Spee, J. I. de Vicente, and B. Kraus,
J. Math. Phys. {\bf 57}, 052201 (2016).

\bibitem{trivialsym} 
D. Sauerwein, N. R. Wallach, G. Gour, and B. Kraus, 
Phys. Rev. X \textbf{8}, 031020 (2018).


\bibitem{SLOCC222} 
W. D\"ur, G. Vidal, and J. I. Cirac,
Phys. Rev. A \textbf{62}, 062314 (2000).

\bibitem{SLOCC2222} 
F. Verstraete, J. Dehaene, B. De Moor, and H. Verschelde, 
Phys. Rev. A \textbf{65}, 052112 (2002).

\bibitem{SLOCC333} 
E. Briand, J. Luque, J. Thibon, and F. Verstraete, 
J. Math. Phys. \textbf{45}, 4855 (2004).

\bibitem{bastin}
 T. Bastin, S. Krins, P. Mathonet, M. Godefroid, L. Lamata, and E. Solano,
 Phys. Rev. Lett. {\bf 103}, 070503 (2009).
 
\bibitem{ChMi10} E. Chitambar, C.A. Miller, and Y. Shi, 
J. Math. Phys. \textbf{51}, 072205 (2010).

\bibitem{sanz}
M. Sanz, D. Braak, E. Solano, and I. L. Egusquiza,
J. Phys. A: Math. Theor. {\bf 50}, 195303 (2017).

\bibitem{hebenstreit}
M. Hebenstreit, M. Gachechiladze, O. Gühne, and B. Kraus,
Phys. Rev. A {\bf 97}, 032330 (2018). 

\bibitem{liduan}
Y. Li, Y. Qiao, X. Wang, and R. Duan,
Commun. Math. Phys. {\bf 358}, 791 (2018).

\bibitem{sawicki}
A. Sawicki, T. Maci{a}\.zek, M. Oszmaniec, 
K. Karnas, K. Kowalczyk-Murynka, and M. Ku\'s,
Rep. Math. Phys. {\bf 82}, 81 (2018).

\bibitem{steinhoff}
F.E.S. Steinhoff, 
arXiv:1905.01824

\bibitem{SLOCCfinite} 
L. Chen, Y.-X. Chen, and Y.-X. Mei, Phys. Rev. A \textbf{74}, 052331 (2006).

\bibitem{peresppt}
A. Peres, Phys. Rev. Lett. {\bf 77}, 1413 (1996).

\bibitem{sep1}  
P. Horodecki, Phys. Lett. A \textbf{232}, 333 (1997).

\bibitem{sep2} 
K. Chen and L.-A. Wu, 
Quantum Inf. Comput. \textbf{39}, 193 (2003).

\bibitem{rudolph}
O. Rudolph,
Quantum Inf. Proc. {\bf 4}, 219 (2005).

\bibitem{doherty}
A. C. Doherty, P. A. Parrilo, and  F. M. Spedalieri
Phys. Rev. A {\bf 71}, 032333 (2005).

\bibitem{sep3} 
J. I. de Vicente, 
Quantum Inf. Comput. \textbf{7}, 624 (2007).

\bibitem{sepNP} 
L. Gurvits, in Proc. of the 35th ACM Symp. on Theory
of Comp. (ACM Press, New York, 2003), pp. 10-19, 
see also quant-ph/0303055.

\bibitem{gharibian}
S. Gharibian, 
Quantum Inf. Comput. {\bf 10}, 343 (2010).

\bibitem{EW} 
M. Horodecki, P. Horodecki, and R. Horodecki, 
Phys. Lett. A \textbf{223}, 1 (1996).

\bibitem{EW1} 
B. M. Terhal, 
Phys. Lett. A \textbf{271}, 319 (2000).

\bibitem{3qubitmixed} 
A. Ac\'in, D. Bru{\ss}, M. Lewenstein, and A. Sanpera, 
Phys. Rev. Lett. \textbf{87}, 040401 (2001).


\bibitem{symmetricmixedSLOCC} 
T. Bastin, P. Mathonet, and E. Solano, 
Phys. Rev. A \textbf{91}, 022310 (2015).


\bibitem{hulpke} 
F. Hulpke, D. Bru{\ss}, M. Lewenstein, and A. Sanpera, 
Quantum Inf. Comput. \textbf{4}, 207 (2004).

\bibitem{gezaspin}
G. T\'oth, 
Phys. Rev. A {\bf 71}, 010301(R) (2005).

\bibitem{hermann}
P. Krammer, H. Kampermann, D. Bru{\ss}, R. A. Bertlmann, L. C. Kwek, 
and C. Macchiavello,
Phys. Rev. Lett. {\bf 103}, 100502 (2009).

\bibitem{kampermann} 
H. Kampermann, O. Gühne, C. Wilmott, and D. Bru{\ss},
Phys. Rev. A {\bf 86}, 032307 (2012).

\bibitem{kraus2010}
B. Kraus, Phys. Rev. A {\bf 82}, 032121 (2010).

\bibitem{TensProdEW} 
A. Rutkowski and  P. Horodecki, 
Phys. Lett. A \textbf{378}, 2043 (2014).

\bibitem{sperling}
F. Shahandeh, J. Sperling, and W. Vogel,
Phys. Rev. Lett. {\bf 113}, 260502 (2014).

\bibitem{SLOCCNP} 
E. Chitambar, R. Duan, and Y. Shi, 
Phys. Rev. Lett. \textbf{101}, 140502 (2008).

\bibitem{SDP} 
L. Vandenberghe and S. Boyd, 
SIAM Rev. \textbf{38}, 49 (1996).

\bibitem{pianistates} 
M. Piani and C. Mora, 
Phys. Rev. A \textbf{75}, 012305 (2007).


\end{thebibliography}
\end{document}